\documentclass[runningheads]{llncs}
\usepackage[T1]{fontenc}
\usepackage{graphicx}
\usepackage{natbib}
\usepackage[utf8]{inputenc} 
\usepackage{hyperref}       
\usepackage{url}            
\usepackage{booktabs}       
\usepackage{amsfonts}       
\usepackage{tikz}           
\usepackage{nicefrac}       
\usepackage{microtype}      
\usepackage{xcolor}         

\usepackage{graphicx}
\usepackage{subfigure}
\usepackage{subcaption}

\usepackage{amsmath}
\usepackage{amssymb}
\usepackage{mathtools}

\usepackage[capitalize,noabbrev]{cleveref}

\usepackage{pgfplots}
\usepackage{bbm}
\pgfplotsset{compat=1.18}
\usepackage{float}

\spnewtheorem{thm}{Theorem}{\bfseries}{\itshape}
\spnewtheorem{cor}[thm]{Corollary}{\bfseries}{\itshape}
\spnewtheorem{lem}[thm]{Lemma}{\bfseries}{\itshape}
\spnewtheorem{prop}[thm]{Proposition}{\bfseries}{\itshape}
\spnewtheorem{conj}[thm]{Conjecture}{\bfseries}{\itshape}
\spnewtheorem{obs}[thm]{Observation}{\bfseries}{\itshape}
\spnewtheorem{clm}[thm]{Claim}{\bfseries}{\itshape}

\spnewtheorem{algo}{Algorithm}{\bfseries}{\itshape}

\spnewtheorem{df}[thm]{Definition}{\bfseries}{\rmfamily}
\spnewtheorem{eg}[thm]{Example}{\bfseries}{\rmfamily}
\spnewtheorem{asm}[thm]{Assumption}{\bfseries}{\rmfamily}
\spnewtheorem{cond}[thm]{Condition}{\bfseries}{\rmfamily}

\spnewtheorem{rmk}[thm]{Remark}{\itshape}{\rmfamily}

\newcommand{\NEG}{\mathrm{pure NE}(G)}
\newcommand{\R}{\mathbb{R}}
\renewcommand{\t}{\mathbf{t}}
\newcommand{\x}{\mathbf{x}}
\newcommand{\X}{\mathcal{X}}
\newcommand{\y}{\mathbf{y}}
\newcommand{\z}{\mathbf{z}}
\newcommand{\argmin}{\mathop{\mathrm{argmin}}}
\newcommand{\argmax}{\mathop{\mathrm{argmax}}}

\allowdisplaybreaks

\begin{document}

\title{Pure Nash Equilibria under the Affine Mechanism: A Potential Game of Exaggeration}
%
\titlerunning{Pure Nash Equilibria under the Affine Mechanism}
%
\author{Jason Jisen Li \and Young Wu \and Yancheng Zhu \and Jin-yi Cai \and Xiaojin Zhu}
\authorrunning{Jason Jisen Li et al.}
%
\institute{University of Wisconsin - Madison}
%
\maketitle              
\begin{abstract}
The mean mechanism is known to be non-incentive-compatible, namely, rational players are incentivized to misreport their values.
Despite this game-theoretic issue, the mean mechanism is prevalent in practice due to its other desirable properties.
We give a full characterization of pure Nash equilibria--how the players will misreport--for the affine mechanism, of which the mean is a special case.
Furthermore, we characterize both complete-information and Bayesian games under the affine mechanism.
Our results highlight the inevitability of extreme exaggeration in such games.


\end{abstract}

\section{Introduction} 
Consider the classic numerical aggregation problem: $n$ players each report a number $x_1 \ldots x_n$. A mechanism aggregates these into a single number $\hat x$, e.g., using the mean or the median function. Player $i$'s utility depends on how close $x_i$ and $\hat x$ are (assuming the reported $x_i$ is truthful).
There has been a sharp divide in theory and practice:
\begin{itemize}
\item In theory, it is well understood that on single-peaked preference domains, the strategyproof rules are exactly the generalized median (phantom-voter) aggregation mechanisms~(\citep{moulin1980strategy}). 
In contrast, the mean mechanism is the textbook example of non-incentive-compatibility: players will have an incentive to misreport their numbers (\cite{procaccia2013approximate,chan2021mechanism}).

\item In practice, the mean remains the dominant aggregation rule when participants submit numerical reports.
Corporate and academic performance evaluations such as paper reviews or student evaluations of teaching are almost universally aggregated with the mean.
In congressional committee legislative voting over continuous variables (e.g., deciding the specific dollar amount for a corporate tax rate), final compromises are routinely arrived at via split-the-difference arithmetic averages (\cite{martin2001congressional}).
Many expert judgment procedures, such as engineering risk assessment or technology forecasting, aggregate experts by equal-weight averaging (\cite{cooke1991experts}). 
\end{itemize}

Using the mean is often not out of ignorance: it is a conscious choice to trade strategyproofness for better estimates under the L2 utilitarian social welfare function (\cite{gershkov2017optimal}),
for fairness when the policymaker believes extreme outliers deserve a voice even if it invites cheating (\cite{maskin2008mechanism}), 
and for interpretability, where explaining a median mechanism with phantom agents to the public is politically unpalatable (\cite{walshequitable}). 
When the stake is high, often practitioners patch up the mean mechanism via trimmed mean (as in skating judging) or weighted average, rather than completely abandoning the mean mechanism for the median.
Therefore, despite being non-strategyproof, the mean mechanism is practiced ubiquitously. 

While it was well-known that strategic players \textbf{will} misreport under the mean mechanism, it was less clear \textbf{how} they misreport.
A better understanding is important for detecting strategic behaviors, and for making an informed decision when choosing a mechanism to balance welfare efficiency and control cheating.
When players' strategic misreporting is viewed as an attempt to ``game the system'', it also has security implications.
This paper gives a full characterization of how the players rationally misreport under affine mechanisms, including the mean mechanism.
Our main contributions are:
\begin{enumerate}
\item We prove that pure Nash Equilibria (pure NEs) always exist under the mean mechanism.  This is done via identifying its potential function.
\item In each pure NE, all but at most one player must exaggerate their misreporting to the maximum extent possible. 
\item In decomposable high-dimensional action spaces, the pure NEs are also decomposable into individual 1D problems and admit closed-form solutions.
\item When players do not know other players' types, we provide a similar Bayesian NE characterization.
\end{enumerate}

\section {Related Work}

Our work is closely related to the facility location problem. \cite{moulin1980strategy} and \cite{procaccia2013approximate} introduce the mechanism design problem for facility location, where the mechanism designer selects a facility location as a function of the players' reported preferences in a way that rational players will truthfully report their preferences. In particular, they note that direct averaging of players' reports will incentivize them to misreport. \cite{berga2002single} extends the problem to a multidimensional domain similar to our setting, but it focuses on an incentive-compatible mechanism called the generalized median voter scheme. 
Our paper is not to advocate for the mean mechanism. 
But given its practical relevance, we focus on characterizing its pure NEs.
\cite{schummer2002strategy, filimonov2023strategyproof} study the facility location problem on graphs and trees; \cite{goel2020optimality, zampetakis2023bayesian, barak2024mac} study the approximation mechanisms in multidimensional Euclidean spaces, which use similar techniques and are related to our reduction of the multidimensional problem into multiple single-dimensional ones, but they also focus on decomposition of the generalized median type mechanisms instead of the mean mechanism.

An additional justification for the mean mechanism stems from a decision-theoretic perspective: the optimality of the Bayes estimator for the posterior mean. Under quadratic loss and Gaussian distributional assumptions, weighted averaging is not only commonly used but also optimal. This is empirically studied and commonly assumed in the social science literature, including models by \cite{degroot1991optimal, romeijn2024interpretation}, which argue that averaging, or what they call linear opinion pooling, is a form of optimal Bayesian updating based on other people's opinions, assuming truthfulness. They interpret the weights in the averaging process as trust coefficients, which our model also uses. In the economics and forecasting literature, simple averaging often performs better than more complex methods across various economic forecasting tasks, as summarized in \cite{clemen1989combining, wang2023forecast}.

One interesting line of research observed human exaggeration behaviors in the context of political polarization and media bias, including \cite{jost2022cognitive, kubin2021role, van2021social}. Although these studies do not directly link the behavior to a simple mean mechanism, they provide empirical evidence of exaggeration as an equilibrium behavior, and the non-incentive-compatibility of humans' opinion aggregation mechanisms could be an important contributing factor. 

\section{Problem Setup}
We consider the affine mechanism in an $n$-player simultaneous-move game.
The players have a common continuous action space $\X \subset \R^d$ which is compact and convex.
Let $\x_i \in \X$ be the action of the $i$th player for $i \in [n] \coloneqq \{1, 2, ..., n\}$.
A joint action, or pure strategy profile, $\x=(\x_1, \ldots, \x_n)$ denotes the simultaneous action choice of all players.
As in standard literature, we also write $\x=(\x_i, \x_{-i})$ when we want to emphasize player $i$.

The mean mechanism aggregates to various degrees the inputs it receives from all players. 
In this paper, we consider a slight generalization of the mean, namely the affine mechanism of the form
\begin{equation}
\label{eq:receiver}
\hat\x := w_0 \x_0 + \sum_{i=1}^n w_i \x_i,
\end{equation}
where $w_i > 0, i \in [n]$ is a real-valued (not necessarily normalized) weight that signifies how much influence player $i$ has on the mechanism.
$\x_0 \in \X$ is a bias term that, together with $w_0\in \R$, denotes a fixed, constant ``background'' influence that is beyond the control of the $n$ players.
The affine mechanism~\eqref{eq:receiver} is common knowledge to all players.

The $n$ players each have a loss (negative utility) function $\ell_i$ defined by a target $\t_i \in \R^d$.
Note the targets are not restricted to $\X$.
The goal of player $i$ is to drive the mechanism's $\hat\x$ close to $\t_i$.
We consider the squared 2-norm loss:
\begin{equation}
\label{eq:ell}
\ell_i(\x) := \|\hat\x - \t_i\|_2^2 = \|w_0 \x_0 + \sum_{j=1}^n w_j \x_j - \t_i\|_2^2.
\end{equation}
As rational agents, the players want to selfishly minimize their own $\ell_i(\x)$.
The players compete with each other when their targets $\t_1, \ldots, \t_n$ are distinct (though our setup allows some or all targets to overlap, too).
The fact that the mechanism's $\hat\x$ is defined by the joint action $\x$ couples the players together in a general-sum game.

The above narrative can be abstracted into the following formal affine mechanism game,
\begin{df} [Affine Mechanism Game (AMG)] \label{df:game}
The Affine Mechanism Game is an $n$-player general-sum game $G = \left(n, \X, \{\ell_i\}_{i=1}^{n}\right)$, 
where $n$ is the number of players, $\X\subset \R^d$ is a compact and convex action space, and $\ell_i: \X^n \mapsto \R$ taking the form of equation~\eqref{eq:ell} is player $i$'s loss function.
\end{df}


In the rest of this section, 
we assume that the targets (and hence loss functions) are common knowledge to all players.
We show $G$ is a potential game and study its pure strategy Nash equilibria.
We then characterize the properties of these pure NEs.
In section~\ref{sec:BG} we relax this assumption and study the resulting Bayesian game.

\subsection{AMG is a Potential Game}
\label{sec:potential}
\begin{df}  \label{df:ne} 
A pure strategy Nash equilibrium of the game $G = \left(n, \X, \left\{\ell_i\right\}_{i=1}^{n}\right)$ is a strategy profile $\x\in \X^n$ satisfying
$\ell_{i}\left(\x_{i}, \x_{-i}\right) \le \ell_{i}\left(\y, \x_{-i}\right), \forall\; \y \in \X, i \in [n]$ .
\end{df}
Our AMG satisfies the Debreu-Glicksberg-Fan theorem and pure NEs exist.
To find the pure NEs explicitly, we turn to potential games. Potential games are games with a special structure that allow strong results on pure Nash equilibria (\cite{monderer1996potential}).
We now show AMG is a potential game:

\begin{thm} [Potential Game] \label{thm:ctp} 
The AMG $G$ is a potential game with the potential function
\begin{equation}
\phi(\x_1,\ldots,\x_n) := \left\| \sum_{i=0}^n w_i \x_i \right\|_2^2 - 2 \sum_{i=1}^n w_i \t_i^\top \x_i.
\label{eq:potential}
\end{equation}
\end{thm}

As a potential game, AMG has pure NEs.  Section~\ref{sec:structure} will be dedicated to the structure of these pure NEs.
First, we remark on algorithms for finding pure NEs of AMG.
Due to Proposition~\ref{prop:cvx}, any convex optimization algorithm that minimizes the convex function $\phi$ over the convex set $\X^n$ with strong guarantees can be utilized to find a pure NE~(\cite{boyd2004convex}).
Meanwhile, in the game theory community, the best-response dynamics is a traditional algorithm for finding a pure NE in potential games~(\cite{roughgarden2010algorithmic}): 
\begin{rmk} [Best-Response Dynamics] \label{rmk:br} 
Starting from an arbitrary $\x^{\left(0\right)} \in \X^n$, the best response sequence $\left\{\x^{\left(t\right)}\right\}_{t=1}^{\infty}$, under standard coordinate descent conditions (\cite{tseng2001convergence}), every limit point of the sequence is a minimizer of $\phi$, hence a pure Nash equilibrium of $G$, where
$\mathbf{x}^{\left(t\right)}_{i} = \mathop{\mathrm{argmin}}_{\mathbf{x}_{i}} \ell_i\left(\mathbf{x}_{i}, \mathbf{x}^{\left(t - 1\right)}_{-i}\right)$ if $i = ((t-1) \mod n) + 1$, and $\mathbf{x}^{\left(t-1\right)}_{i}$ otherwise.
\end{rmk}

One interesting observation that is relevant for AMG is that, under best response dynamics, no player needs to know other players' targets.
Concretely, the players may carry out the best response dynamics as a learning dynamics in a distributed fashion over time, with no two players simultaneously updating their actions.
When player $i$ updates its own action $\x_i$, it observes other player's most recent actions $\x_{-i}$ but does not need to know their targets $\t_{-i}$.
Therefore, the best response dynamics may offer a computational account of how players in the real world iteratively adjust their actions based on actions of other players without knowing the other players' true intentions, and still reach an equilibrium.

Note that minimizing the potential function $\phi$ along the $\x_i$ direction is equivalent to minimizing its own loss function $\ell_i$:
$\mathop{\mathrm{argmin}}_{\x_i} \phi(\x_i, \x_{-i})=
\mathop{\mathrm{argmin}}_{\x_i} \ell_{i}(\x_i, \x_{-i})$. As a result, best response dynamics also correspond to cyclic coordinate descent on $\phi$ (~\cite{wright2022optimization}).

\subsection{Structure of Pure Nash Equilibria in AMG}
\label{sec:structure}

We next show that $\x\in\X^n$ is a pure NE of $G$ if and only if $\x$ is a minimum of $\phi$ restricted to the domain $\X^n$. 

\begin{prop} [Pure NEs $\iff$ minima] \label{prop:cvx} 
The set of pure Nash equilibria in $G$ is
$\NEG = \mathop{\mathrm{argmin}}_{\x\in\X^n} \phi(\x)$,
and moreover, the set of correlated equilibria is the set of mixtures of pure Nash equilibria, that is,
$\mathrm{CE}(G) = \Delta\NEG$.
\end{prop}


We remark that by definition $\X$ is compact and convex, thus $\X^n$ is bounded and closed.  
The potential function $\phi(\x)$ may not have a global minimum on the extended domain $\R^{nd}$ (it could diverge to $-\infty$ there), but on $\X^n$ it will have at least one minimum (perhaps on the boundary).
In fact, we have the following guarantee.

\begin{cor} [Cardinality of $\NEG$] \label{cor:cvx} $\NEG$ is non-empty and convex, that is, $G$ has either one pure NE or infinite pure NEs.
\end{cor}

In particular, the second part of Proposition~\ref{prop:cvx} and this Corollary imply that in the case of a unique pure strategy Nash equilibrium $\x$ (for example, in the 1D two-player case with distinct targets, which we expand on in the next subsection), then $\x$ is the unique Nash equilibrium and the unique correlated equilibrium.

We provide a few illustrative examples of AMG in subsection~\ref{sec:pure}. In AMG the pure NEs often seem to involve \emph{all} players misrepresenting their true target to the mechanism (i.e. $\x_i \neq \t_i$).
Furthermore, such misrepresentation often seems to take the form of \emph{extreme exaggeration}, in the sense that a player's rational action $\x_i$ at any pure NE is often pushed to the boundary of the action space $\X$ so they cannot exaggerate the action further.
These observations are almost true, but need a subtle correction: 
Given that the players' targets $\t_1 \ldots \t_n$ are all distinct, \textbf{at any pure NE, extreme exaggeration is necessary for all (except perhaps one) players}.
In other words, there are two possibilities:
Either all players are at the boundary of $\X$, or one of them is in the interior of $\X$.
In the latter case, that special player (say player $i^*$) will in general still need to exaggerate its target $\x_{i^*} \neq \t_{i^*}$ to the mechanism; it is just that $\x_{i^*}$ is in the interior of $\X$ and not at the boundary.
Furthermore, player $i^*$ is the lucky winner
in that the mechanism will end up at its target $\t_{i^*}$. 
Our next theorem precisely characterizes this phenomenon. 


\begin{thm} [All-But-At-Most-One Extreme Exaggeration] \label{thm:bdy} If $\left\{\mathbf{t}_{i}\right\}_{i=1}^{n}$ are all distinct, then every pure NE satisfies the property that
$\left| \left\{i\in[n] : \mathbf{x}_{i} \in \mathrm{int\;} \X\right\} \right| \leq 1$. 
Conversely, in a NE, if $\x_i, \x_j \in \mathrm{int\;} \X$, then $\t_i = \t_j$. Furthermore, if $\x_{i^\star} \in \mathrm{int\;} \X$ for some $i^\star \in \left[n\right]$, then
$\hat{\mathbf{x}}= \mathbf{t}_{i^\star}$.
\end{thm}

\subsection{Computing Pure Nash Equilibria}
\label{sec:pure}

We now show how to find pure NEs in high dimensional action space by decomposing the problem.
The base case is when $\X$ is in a one-dimensional space, under the assumption that it is compact and convex, we can write $\X = \left[L, U\right]$, and in this case, we can fully characterize the set of all Nash equilibria as follows.

\begin{lem} [Nash Equilibria in $1D$] \label{lem:odne} 
When $\X = \left[L, U\right]$, then there exists some $t^\star \in \mathbb{R}$, such that in every Nash equilibrium $\left(x_{1}, x_{2}, ..., x_{n}\right)$ of AMG, we have
\begin{equation} \begin{aligned}
\label{eq:odne}
&x : \begin{cases} x_{i} = L & \text{\;if\;} t_{i} < t^\star \\ x_{i} :  \displaystyle\sum_{i : t_{i} = t^\star} w_{i} x_{i} = t^\star - w_{0} x_{0} - \displaystyle\sum_{i : t_{i} < t^\star} w_{i} L - \displaystyle\sum_{i : t_{i} > t^\star} w_{i} U & \text{\;if\;} t_{i} = t^\star \\ x_{i} = U & \text{\;if\;} t_{i} > t^\star \\ \end{cases} .
\end{aligned} \end{equation}\end{lem}
In particular, when targets are sorted $t_{1} \leq t_{2} \leq ... \leq t_{n}$, ~\eqref{eq:odne} is equivalent to the existence of $\underline{i}$ and $\overline{i}$ such that in every NE $x_{1} = x_{2} = ... = x_{\underline{i}} = L$, $x_{\overline{i}} = x_{\overline{i} + 1} = ... = x_{n} = U$, and if there are players with the target $t_{i} = t^\star$ with $\underline{i} < i < \overline{i}$, then it is NE as long as their $x_{i}$ results in $\hat{x} = t^\star$.
\begin{rmk}  \label{rmk:odalg} 
A simple algorithm to find an NE is to iterate over the following set, $t^\star \in \left\{t_{1}, t_{2}, ..., t_{n}\right\} \cup \left\{\dfrac{1}{2} \left(t_{1} + t_{2}\right), \dfrac{1}{2} \left(t_{2} + t_{3}\right), ..., \dfrac{1}{2} \left(t_{n-1} + t_{n}\right)\right\} \cup \left\{L, U\right\}$, and for each candidate $t^\star$, compare the value of the potential function $\phi\left(x\right)$ with $x$ satisfying ~\eqref{eq:odne}, and the $x$ that minimizes $\phi$ is a NE.

\end{rmk}

In general, it is difficult to find a closed form solution in higher dimensions, but for special $\X$ that is a Cartesian product of intervals (i.e. a hyper-rectangle), that is $\X = \displaystyle\prod_{k=1}^{d} \X_{k}$ with $\X_{k} = \left[L_{k}, U_{k}\right]$, we provide the following decomposition result, which implies an algorithm to solve $d$ one-dimensional problems using Lemma~\ref{lem:odne} to obtain the NE coordinate-wise.

Since the loss function is coordinate-wise-separable, $\ell_i(\x) = \displaystyle\sum_{k=1}^d \ell_{ik}(\x_k)$, where $\ell_{ik}(\x_k) := (\hat{\x}_k-\t_{ik})^2$. 
We can define for each $k\in[d]$, $G_k$ with the action space $\X_k = [L_k, U_k]$, loss $\ell_{ik}$, and mechanism $\hat{\x}_k = w_0\x_{0k} + \displaystyle\sum_{i=1}^n w_i\x_{ik}$, for any player $i \in [n]$. Then, the pure Nash Equilibrium of the multi-dimensional game is simply a cross-product of all $d$ $\mathrm{pure NE}(G_k)$ sets. More formally, a conversion between the Nash equilibrium in $G$ and the 1D games $G_k$ can be established through the following proposition:

\begin{prop} [Decomposition of Multi-dimensional AMG] \label{prop:dec} 
Assume $\X=\displaystyle\prod_{k=1}^d \X_{k}\subseteq\R^d$, then the set of pure Nash equilibria of $G$ can be characterized by $\mathrm{pure NE}(G) = \displaystyle\prod_{k=1}^d \mathrm{pure NE}(G_k)$, where $G_k = (n, \X_k, \ell_{ik})$.
\end{prop}

Given this proposition, it is sufficient to form and solve the $d$ one-dimensional games $G_k$ independently, and take the Cartesian product of their pure Nash equilibrium sets to obtain $\mathrm{pure NE}(G)$ for multi-dimensional AMG $G$. 

\begin{eg} [NEs for Two-Player AMG in 1D and 2D]
In a one-dimensional two-player AMG $G$ with targets $t_1$ and $t_2$, pure actions $x_1, x_2 \in \X = [L, U]$, and mechanism $\hat{x} = w_0 x_0 + w_1 x_1 + w_2 x_2$, the pure NE set only depends on the relative positions of $t_1$ and $t_2$ with respect to the weighted middle point $m$, 
where $m = w_0x_0 + w_1L + w_2U$ if $t_1 \le t_2$, and $m = w_0x_0 + w_1U + w_2L$ if $t_1 > t_2$. The closed form solution of pure NE can be divided into three cases:
\begin{enumerate}
\item If $t_1=t_2=t$:
$
\mathrm{pure NE}(G) =
\begin{cases}
\{(L, L)\}, & \text{if}\; t < w_0x_0 + L(w_1 + w_2), \\
\{(U, U)\}, & \text{if}\; t > w_0x_0 + U(w_1 + w_2), \\
\{(x_1,x_2) : \hat{x}=t\}, & \text{otherwise}.
\end{cases}
$
\item If $t_1 \ne t_2$ and $t_1 < m < t_2$ or $t_2 < m < t_1$:
$
\mathrm{pure NE}(G) = 
\begin{cases}
\{(L, U)\}, & \text{if}\; t_1 < t_2, \\
\{(U, L)\}, & \text{if}\; t_2 < t_1.
\end{cases}
$
\item If $t_1 \ne t_2$ and $t_1, t_2 \le m$ or $t_1, t_2 \ge m$:

$
\mathrm{pure NE}(G) = 
\begin{cases}
\{(L, \max(L, \frac{1}{w_2} (t_2 - w_0x_0 - w_1L)))\}, & \text{if}\; t_1 < t_2 \le m, \\
\{(\max(L, \frac{1}{w_1} (t_1 - w_0x_0 - w_2L)), L)\}, & \text{if}\; t_2 < t_1 \le m, \\
\{(U, \min(U, \frac{1}{w_2} (t_2 - w_0x_0 - w_1U)))\}, & \text{if}\; m \le t_2 < t_1, \\
\{(\min(U, \frac{1}{w_1} (t_1 - w_0x_0 - w_2U)), U)\}, & \text{if}\; m \le t_1 < t_2.
\end{cases}
$
\end{enumerate}

In case 1, the two targets coincide. NE happens when $\hat{x}$ is in the location closest to $t$. 
This is also the case when there can be infinite many pure NEs (the third branch). In case 2, the targets lie on different sides of m. Both players must exaggerate their actions to the maximum extent possible. In case 3, both targets lie on the same side of $m$. One player will exaggerate to the maximum extent, while the other either ``wins the game'' by forcing the mechanism to land on its own target or also exaggerates to the maximum when winning is impossible due to a large $w_0x_0$.

We illustrate in Figure~\ref{fig:1dne} the three NE cases for 1-dimensional two player AMG examples with domain $\X = [L, U]$, equal weights $w_1=w_2=\frac{1}{2}$, and bias term $w_0x_0=0$. In case 1, $x_1$ and $x_2$ can be anywhere in the gray box, as long as $\hat{x}$ sits on $t_1=t_2$. Thus, there are infinite pure NE. Case 2 and 3 both have a unique pure NE. In case 2 $x_1$ and $x_2$ both exaggerate to the extreme action, $L$ or $U$, that is closest to their respective targets. In case 3, player 1 plays an interior action and ``wins'' as $\hat{x}$ lands on $t_1$, while player 2 chooses an extreme action.
 
\begin{figure}
\centering
\begin{tikzpicture}[x=1pt,y=1pt,scale=0.75]
\begin{scope}[yshift=110pt]
\fill[fill=lightgray] (150,60) rectangle (450,40);
\draw[black,line width=1pt, |-|] (50,50) -- (450,50);
\fill[fill=orange] (295,45) rectangle (305,55);
\node[above, text=black] at (300,55) {$\hat{x}$};
\node[below, text=black] at (290,45) {$t_1$};
\node[below, text=black] at (310,45) {$t_2$};
\fill[fill=blue] (200,50) circle[radius=5];
\fill[fill=green] (400,50) circle[radius=5];
\node[above, text=black] at (200,55) {$x_1$};
\node[above, text=black] at (400,55) {$x_2$};
\node[below] at (250, 40) {Case 1};
\end{scope}

\begin{scope}[yshift=55pt]
\draw[black,line width=1pt, |-|] (50,50) -- (450,50);
\fill[fill=orange] (245,45) rectangle (255,55);
\fill[fill=blue] (95,45) rectangle (105,55);
\fill[fill=green] (345,45) rectangle (355,55);
\node[above, text=black] at (250,55) {$\hat{x}$};
\node[above, text=black] at (100,55) {$t_1$};
\node[above, text=black] at (350,55) {$t_2$};
\fill[fill=blue] (50,50) circle[radius=5];
\fill[fill=green] (450,50) circle[radius=5];
\node[above, text=black] at (50,55) {$x_1$};
\node[above, text=black] at (450,55) {$x_2$};
\node[below] at (250, 40) {Case 2};
\end{scope}

\begin{scope}[yshift=0pt]
\draw[black,line width=1pt, |-|] (50,50) -- (450,50);
\fill[fill=orange] (295,45) rectangle (305,55);
\fill[fill=green] (395,45) rectangle (405,55);
\node[above, text=black] at (300,55) {$\hat{x}$};
\node[below, text=black] at (295,45) {$t_1$};
\node[above, text=black] at (400,55) {$t_2$};
\fill[fill=blue] (150,50) circle[radius=5];
\fill[fill=green] (450,50) circle[radius=5];
\node[above, text=black] at (150,55) {$x_1$};
\node[above, text=black] at (450,55) {$x_2$};
\node[below] at (250, 40) {Case 3};
\end{scope}

\end{tikzpicture}
\caption{Pure NE of two-player 1D AMG with Equal Weights} \label{fig:1dne}
\end{figure}
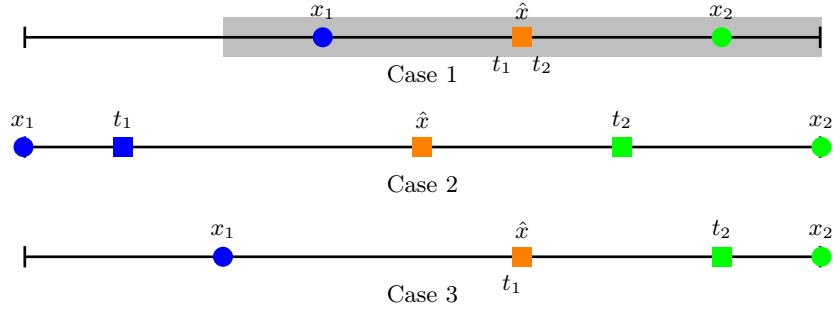

Now, for multi-dimensional games we can combine the 1D pure NEs.
In particular, for a 2-dimensional game, the pure NE falls into six categories formed by combining the three 1-D characterizations, which yields $3\times3=9$ cases, then eliminating the three mirroring cases. In Figure~\ref{fig:tdne} we present these six characterizations of a 2-dimensional AMG with $\X = [L, U]^2$, uniform weights $w_1 = w_2 = \frac{1}{2}$, and $w_0x_0 = 0$. These derive directly from combining the three 1-dimensional NE cases in Figure~\ref{fig:1dne}.
In case 1 ($\t_1=\t_2$), $\x_1, \x_2$ can be anywhere in the gray box so long as they average to $\t_1$. 
In cases 2 and 3, $\t_1,\t_2$ agree on one coordinate, and the infinite many pure NEs are chosen from the blue line and the green line, respectively.
Cases 4,5,6 have a unique pure NE.
Case 6 is where player 1 is the ``winner'' and has an interior action.

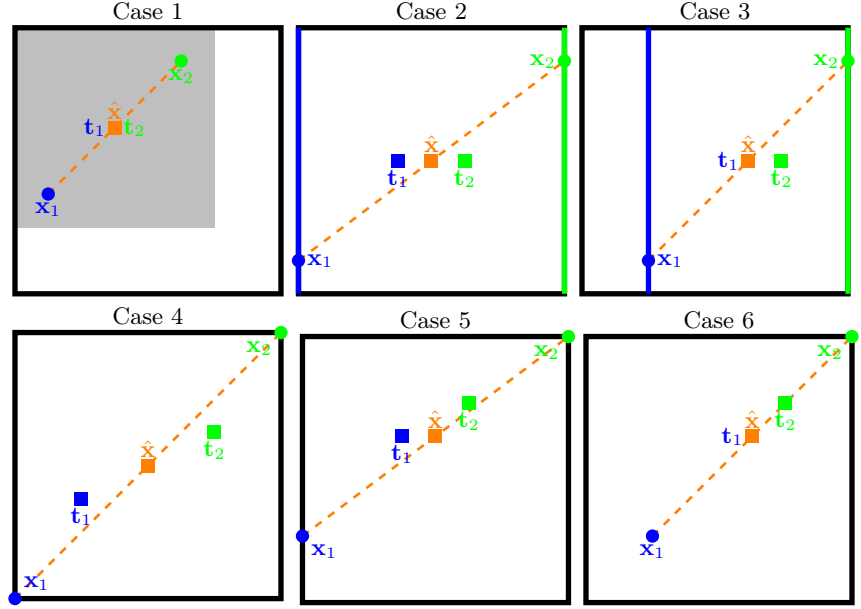
\begin{figure}
\centering
\begin{tikzpicture}[x=1pt,y=1pt,scale=0.25]
\fill[fill=lightgray] (50,450) rectangle (350,150);
\draw[black, line width=2pt] (50,450) rectangle (450,50);
\draw[orange, line width=1pt, dashed] (100,200) -- (300,400);
\fill[fill=orange] (190,290) rectangle (210,310);
\node[above, text=orange] at (200,300) {$\hat{\x}$};
\fill[fill=blue] (100,200) circle[radius=10];
\node[below, text=blue] at (100,200) {$\x_1$};
\fill[fill=green] (300,400) circle[radius=10];
\node[below, text=green] at (300,400) {$\x_2$};
\node[left, text=blue] at (200,300) {$\t_1$};
\node[right, text=green] at (200,300) {$\t_2$};
\node[above] at (250, 450) {Case 1};
\end{tikzpicture}
\begin{tikzpicture}[x=1pt,y=1pt,scale=0.25]
\draw[black, line width=2pt] (50,450) rectangle (450,50);
\fill[fill=orange] (240,240) rectangle (260,260);
\node[above, text=orange] at (250, 250) {$\hat{\x}$};
\draw[orange, line width=1pt, dashed] (50,100) -- (450,400);
\draw[blue, line width=2pt] (50,50) -- (50,450);
\fill[fill=blue] (50,100) circle[radius=10];
\node[right, text=blue] at (50,100) {$\x_1$};
\draw[green, line width=2pt] (450,50) -- (450,450);
\fill[fill=green] (450,400) circle[radius=10];
\node[left, text=green] at (450,400) {$\x_2$};
\fill[fill=blue] (190,240) rectangle (210,260);
\node[below, text=blue] at (200,250) {$\t_1$};
\fill[fill=green] (290,240) rectangle (310,260);
\node[below, text=green] at (300,250) {$\t_2$};
\node[above] at (250, 450) {Case 2};
\end{tikzpicture}
\begin{tikzpicture}[x=1pt,y=1pt,scale=0.25]
\draw[black, line width=2pt] (50,450) rectangle (450,50);
\fill[fill=orange] (290,240) rectangle (310,260);
\node[above, text=orange] at (300, 250) {$\hat{\x}$};
\draw[orange, line width=1pt, dashed] (150,100) -- (450,400);
\draw[blue, line width=2pt] (150,50) -- (150,450);
\fill[fill=blue] (150,100) circle[radius=10];
\node[right, text=blue] at (150,100) {$\x_1$};
\draw[green, line width=2pt] (450,50) -- (450,450);
\fill[fill=green] (450,400) circle[radius=10];
\node[left, text=green] at (450,400) {$\x_2$};
\node[left, text=blue] at (300,250) {$\t_1$};
\fill[fill=green] (340,240) rectangle (360,260);
\node[below, text=green] at (350,250) {$\t_2$};
\node[above] at (250, 450) {Case 3};
\end{tikzpicture}
\begin{tikzpicture}[x=1pt,y=1pt,scale=0.25]
\draw[black, line width=2pt] (50,450) rectangle (450,50);
\fill[fill=orange] (240,240) rectangle (260,260);
\node[above, text=orange] at (250, 250) {$\hat{\x}$};
\draw[orange, line width=1pt, dashed] (50,50) -- (450,450);
\fill[fill=blue] (50,50) circle[radius=10];
\node[above right, text=blue] at (50,50) {$\x_1$};
\fill[fill=green] (450,450) circle[radius=10];
\node[below left, text=green] at (450,450) {$\x_2$};
\fill[fill=blue] (140,190) rectangle (160,210);
\node[below, text=blue] at (150,200) {$\t_1$};
\fill[fill=green] (340,290) rectangle (360,310);
\node[below, text=green] at (350,300) {$\t_2$};
\node[above] at (250, 450) {Case 4};
\end{tikzpicture}
\begin{tikzpicture}[x=1pt,y=1pt,scale=0.25]
\draw[black, line width=2pt] (50,450) rectangle (450,50);
\fill[fill=orange] (240,290) rectangle (260,310);
\node[above, text=orange] at (250, 300) {$\hat{\x}$};
\draw[orange, line width=1pt, dashed] (50,150) -- (450,450);
\fill[fill=blue] (50,150) circle[radius=10];
\node[below right, text=blue] at (50,150) {$\x_1$};
\fill[fill=green] (450,450) circle[radius=10];
\node[below left, text=green] at (450,450) {$\x_2$};
\fill[fill=blue] (190,290) rectangle (210,310);
\node[below, text=blue] at (200,300) {$\t_1$};
\fill[fill=green] (290,340) rectangle (310,360);
\node[below, text=green] at (300,350) {$\t_2$};
\node[above] at (250, 450) {Case 5};
\end{tikzpicture}
\begin{tikzpicture}[x=1pt,y=1pt,scale=0.25]
\draw[black, line width=2pt] (50,450) rectangle (450,50);
\fill[fill=orange] (290,290) rectangle (310,310);
\node[above, text=orange] at (300, 300) {$\hat{\x}$};
\draw[orange, line width=1pt, dashed] (150,150) -- (450,450);
\fill[fill=blue] (150,150) circle[radius=10];
\node[below, text=blue] at (150,150) {$\x_1$};
\fill[fill=green] (450,450) circle[radius=10];
\node[below left, text=green] at (450,450) {$\x_2$};
\node[left, text=blue] at (300,300) {$\t_1$};
\fill[fill=green] (340,340) rectangle (360,360);
\node[below, text=green] at (350,350) {$\t_2$};
\node[above] at (250, 450) {Case 6};
\end{tikzpicture}
\caption{Pure NE of two-player 2D AMG with Equal Weights}
\label{fig:tdne}
\end{figure}
\end{eg}

\section{Bayesian Games under the Affine Mechanism}
\label{sec:BG}
In the previous formulation of AMG, we assumed that each player knows the target location $\t_1, \ldots, \t_n$ of all players, and hence all loss functions are common knowledge. However, it is also interesting to consider the scenario in which each player (say player $i$) only knows its own target $\t_i$, but not any of the other players' target $\t_{-i}$.
For such games with incomplete information, we adopt an interim Bayesian game approach in which the player knows only the prior distributions of the other targets.

Concretely, the Bayesian Affine Mechanism Game has the same formulation as the standard AMG, except that any target $\t_i$ shall be chosen from a distribution $T_i$, for all $i \in [n]$. Each player $i$ knows the exact location of their own target $\t_i$. They also know the distributions of $T_j$, $\forall j \in [n]$, the AMG mechanism, and that other players also know this information. The players then make the simultaneous action $\x_i \in \X$ from the common continuous action space $\X \subset \mathbb{R}^d$ that is compact and convex. All other settings remain the same. We have the following definition:

\begin{equation} \begin{aligned}
\ell_{i}\left(\mathbf{x}_{i}, \mathbf{x}_{-i} | \mathbf{t}_{i}\right) &\coloneqq \left\|\hat{\mathbf{x}} - \mathbf{t}_{i}\right\|_{2}^{2} = \left\|w_{0} \mathbf{x}_{0} + \displaystyle\sum_{j=1}^{n} w_{j} \mathbf{x}_{j} - \mathbf{t}_{i}\right\|_{2}^{2} \label{eq:bell}
\end{aligned} \end{equation}
\begin{df} [Bayesian AMG] \label{df:bg} 
The Bayesian Affine Mechanism Game is an $n$-player general-sum Bayesian game $G = \left(n, \X, \left\{T_{i}\right\}_{i=1}^{n}, \left\{\ell_{i}\right\}_{i=1}^{n}\right)$, where $n$ is the number of players, $\X$ is the compact and convex action space, $T_{i}$ is the independent prior type distribution of player $i$'s target $\mathbf{t}_{i} \in \R^d$ and $\ell_{i} : \X^{n} \times \R^d \to  \mathbb{R}$ is given by~\eqref{eq:bell}.

\end{df}
We denote a player's strategy by $\mathbf{x}_{i} = \pi_{i}\left(\mathbf{t}_{i}\right)$ where $\pi_{i} : \R^d \to  \X$ maps their target position to an action $\mathbf{x}_{i} \in \X$. We again approach the problem with the assumption of a product space domain $\X = \prod_{k=1}^d[L_k, U_k]$, and the following result characterizes the best responses in the Bayesian game.

\begin{prop} [Bayesian Best Responses] \label{prop:bgbr} 
The best response function for each player $i$ in the Bayesian AMG game with $\X = \prod_{k=1}^d[L_k, U_k]$ has the ramp function form,
\begin{equation}\label{eq:bnebr}\begin{aligned}
&\text{br}_{i}\left(\pi_{-i} | t_{i}\right) = \x_i \text{, where for every\;} k \in [d]\\
&x_{ik} = \displaystyle\min\left\{U_k, \displaystyle\max\left\{L_k, \dfrac{1}{w_{i}} t_{ik} - \dfrac{1}{w_{i}} \left(w_0 x_{0k} + \displaystyle\sum_{j \neq i} w_{j} \mathbb{E}_{T_{j}}\left[\pi_{jk}\left(t_{j}\right)\right]\right)\right\}\right\}.
\end{aligned}\end{equation}
In particular, \eqref{eq:bnebr} implies that in all Bayesian Nash equilibria, all players use a strategy with this ramp function form in all coordinates.

\end{prop}
To be specific, in a BNE, the ramp function can be rewritten as, for every player $i$, $\pi_{i}\left(t_{i}\right) = \x_i$, where $x_{ik} = \begin{cases} L_k & \text{\;if\;} t_{ik} \le a_{ik} \\ \dfrac{1}{w_{i}} t_{ik} + c_{ik} & \text{\;if\;} a_{ik} < t_{ik} < b_{ik} \\ U_k & \text{\;if\;} t_{ik} \ge b_{ik} \\ \end{cases} \; \forall k \in [d]$, where $c_{ik} = - \dfrac{1}{w_{i}} \left(w_0 x_{0k} +  \displaystyle\sum_{j \neq i} w_{j} \mathbb{E}_{T_{j}}\left[\pi_{jk}\left(t_{j}\right)\right]\right)$, $a_{ik} = w_{i} \left(L_k - c_{ik}\right)$, and $b_{ik} = a_{ik} + w_{i} \left(U_k - L_k\right) = w_{i} \left(U_k - c_{ik}\right)$. This set of equations in each dimension $k$ characterizes $\left\{c_{ik}\right\}_{i=1}^{n}$ thus $\pi_{i}$. Notice that the equation set in each dimension can be calculated independently.

This result can be interpreted as follows: \textbf{in any dimension, the exaggeration factor is always $\dfrac{1}{w_{i}}$, that is, the less influence the player has, the more the player will exaggerate.}


For arbitrary prior type distributions, $\left\{c_{ik}\right\}_{i=1}^{n}\; \forall k \in [d]$ can be solved numerically, but as an example, we can solve for a closed-form solution under the assumption of uniformly distributed targets for all players.

\begin{eg} [BNE for Uniformly Distributed Targets] \label{eg:1dbne}
A one-dimensional Bayesian AMG with $\X = [0, 1]$, targets $T_i \sim \text{Unif\;}[0, 1]$, bias term $w_0x_0 = 0$, and $\sum_{i=1}^nw_i = 1$ for all players has the unique BNE strategy profile $(\pi_1, \dots, \pi_n)$ given by
\begin{equation} \label{eq:ubne}
\pi_i^*(t_i)=
\begin{cases}
0, & t_i\le \dfrac{1-w_i}{2},\\
\frac{1}{w_i}t_i - \dfrac{1-w_i}{2w_i}, & \dfrac{1-w_i}{2}<t_i < \frac{1+w_i}{2},\\
1, & t_i \ge \dfrac{1+w_i}{2}.
\end{cases}
\end{equation}
\end{eg}

\begin{figure}
\centering
\begin{tikzpicture}[x=1pt,y=1pt,scale=0.2]
\draw[blue, line width=2pt] (50,50) -- (150,50) -- (350,400) -- (450,400);
\draw[black,->, line width=1pt] (50,50) -- (450,50);
\draw[black,->, line width=1pt] (50,50) -- (50,450);
\node[right, text=black] at (450,50) {$t_i$};
\node[above, text=black] at (50,450) {$x_i$};
\node[above] at (250, 500) {$n = 2$ uniform case};
\node[below, text=black] at (150,50) {$\frac{n-1}{2n}$};
\node[below, text=black] at (350,50) {$\frac{n+1}{2n}$};
\draw[black, dashed, line width=1pt] (150,50) -- (150,400);
\draw[black, dashed, line width=1pt] (350,50) -- (350,400);
\node[left, text=black] at (50,50) {$0$};
\node[left, text=black] at (50,400) {$1$};
\end{tikzpicture}
\begin{tikzpicture}[x=1pt,y=1pt,scale=0.2]
\draw[blue, line width=2pt] (50,50) -- (200,50) -- (300,400) -- (450,400);
\draw[black,->, line width=1pt] (50,50) -- (450,50);
\draw[black,->, line width=1pt] (50,50) -- (50,450);
\node[right, text=black] at (450,50) {$t_i$};
\node[above, text=black] at (50,450) {$x_i$};
\node[above] at (250, 500) {$n = 4$ uniform case};
\node[below, text=black] at (200,50) {$\frac{n-1}{2n}$};
\node[below, text=black] at (300,50) {$\frac{n+1}{2n}$};
\draw[black, dashed, line width=1pt] (200,50) -- (200,400);
\draw[black, dashed, line width=1pt] (300,50) -- (300,400);
\node[left, text=black] at (50,50) {$0$};
\node[left, text=black] at (50,400) {$1$};
\end{tikzpicture}
\begin{tikzpicture}[x=1pt,y=1pt,scale=0.2]
\draw[blue, line width=2pt] (50,50) -- (250,50) -- (250,400) -- (450,400);
\draw[black,->, line width=1pt] (50,50) -- (450,50);
\draw[black,->, line width=1pt] (50,50) -- (50,450);
\node[right, text=black] at (450,50) {$t_i$};
\node[above, text=black] at (50,450) {$x_i$};
\node[above] at (250, 500) {$n \rightarrow \infty$ uniform case};
\node[below, text=black] at (250,50) {$\frac{n-1}{2n}, \frac{n+1}{2n}\to \frac12$};
\node[left, text=black] at (50,50) {$0$};
\node[left, text=black] at (50,400) {$1$};
\end{tikzpicture}
\caption{BNE of Games with $n$ Players and Uniform Type Distributions} \label{fig:bnes}
\end{figure}
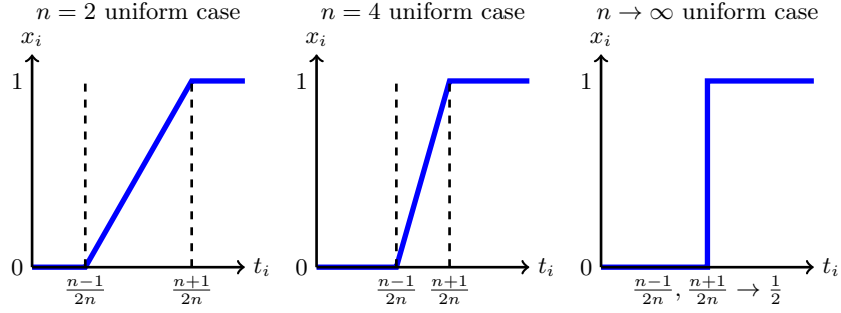

\eqref{eq:ubne} follows from Proposition~\ref{prop:bgbr} by solving the system of $n$ equations in the form of \eqref{eq:bnebr}. More details of the derivation can be found in the appendix. With $w_0 = 0$ and $w_i = \dfrac{1}{n}$ for every $i$, \eqref{eq:ubne} can be further simplified and plotted for various values of $n$ in Figure~\ref{fig:bnes}. In this example, $1 - {(b_i-a_i)} = 1 - w_{i}$ fraction of the types of player $i$ will report one of two extreme types $0$ or $1$, and the remaining $w_{i}$ fraction will exaggerate by a factor of $\dfrac{1}{w_{i}}$.

\section{Extensions of the Affine Mechanism Game}
\subsection{Extension 1: Negative Inner Product Loss}
\label{sec:cosine}
Up to now player $i$'s loss function~\eqref{eq:ell} is based on the Euclidean distance between its target point $\t_i \in \R^d$ and the mechanism $\hat\x$.
In some applications, the following \emph{negative inner product loss} can be more appropriate:
\begin{equation}
\label{eq:ell2}
\ell_i(\x) := - \t_i^\top \hat\x.
\end{equation}
That is, player $i$ has a small loss if the mechanism $\hat\x$ has a large projection onto the direction of target direction $\t_i$.

AMG with this negative inner product loss~\eqref{eq:ell2} has an even stronger guarantee: the game has a \emph{Weakly Dominant Strategy Equilibrium}. 

\begin{df} [Weakly Dominant Strategy Equilibrium (wDSE)] \label{df:sdse} 
An action profile $\x=(\x_1 \ldots \x_n) \in \X^{n}$ is a wDSE if for every player $i$,
$\ell_{i}\left(\x_{i}, \x'_{-i}\right) \leq \ell_{i}\left(\y, \x'_{-i}\right), \; \forall \y \in \X, \x'_{-i} \in \X^{n-1}$.
\end{df}

\begin{rmk}  \label{rmk:ea} 
The term weakly dominant strategy equilibrium is used in Chapter 4.5 of~\cite{tadelis2013game}, and it is also called dominant strategy equilibrium in Chapter 10.3 of~\cite{osborne1994course}, and dominant strategy solution in Chapter 1.3.1 of~\cite{roughgarden2010algorithmic}. As noted in~\cite{osborne1994course}, an action in a wDSE is not required to weakly dominate (with at least one strict inequality) all other actions for a player.
wDSE is also a weaker solution concept compared to (strictly) dominant strategy equilibrium, which requires strict inequality everywhere,
$\ell_{i}\left(\x_{i}, \x'_{-i}\right) < \ell_{i}\left(\y, \x'_{-i}\right), \; \forall \y \in \X, \y  \neq \x_{i}, \x'_{-i} \in \X^{n-1}$.
\end{rmk}

\begin{prop}[Existence of wDSE] \label{prop:wdse}
Under loss~\eqref{eq:ell2}, game $G$ has a wDSE $\x_i^*$ satisfying
$\x_i^* \in \argmax_{\x_i \in \X} \t_i^\top \x_i \;,\forall i\in[n]$.
\end{prop}

One significant benefit of this wDSE is that each player can compute their $\x_i^*$ on their own, without even the knowledge of other players' actions $\x_j, \forall j\neq i$.
Contrast this with the best-response dynamics in section~\ref{sec:potential}.
We remark that $\x_i^*$ may not be exactly along the direction of $\t_i$ since it depends on the domain $\X$.
Nonetheless, computing $\x_i^*$ is a convex optimization problem--maximizing a linear function over a convex set--and thus efficient.


\subsection{Extension 2: Finite Action Space} 
\label{sec:F}
So far, we have assumed that the players' action space $\X$ is a compact and convex (hence infinite unless singleton) subset of $\R^d$. 
In some applications, the players are restricted to picking their actions from a finite $\X$ instead.
For example, $\X$ may be the collection of news articles (represented by embedding vectors in $\R^d$) published by all professional news agencies within the past 24 hours, and each player may select a handful of such news articles to place on a social media user (the mechanism)'s timeline.
This motivates the extension to finite action space:

\begin{df} [AMG with finite action space] \label{df:dbig} 
Affine Mechanism Game with finite action space is an $n$-player general-sum game $F = \left(n, \left\{\mathcal{P}_{k_i}\left(\X\right)\right\}_{i=1}^{n}, \left\{\ell_i\right\}_{i=1}^{n}\right)$, where the action space of player $i$ is $\mathcal{P}_{k_i}\left(\X\right)$, and $\mathcal{P}_{k_i}\left(\X\right)$ is the set of all subsets containing $k_{i}$ elements from a finite $\X \subset \mathbb{R}^{d}$.
The loss function of player $i$
is given by $\ell_{i}\left(\x\right) = \left\|\hat{\x} - \t_{i}\right\|^{2}_{2}$ with
$\hat{\x} = w_{0} \x_{0} + \displaystyle\sum_{i=1}^{n} w_{i} \displaystyle\sum_{j=1}^{k_{i}} \x^{\left(j\right)}_{i}$, 
where $\x^{(j)}_i$ is the $j$th element in player $i$'s chosen subset.
\end{df}

In the original AMG (Definition~\ref{df:game}) where $\X$ is convex, allowing the players to choose multiple items would not affect the results since choosing multiple items is equivalent to choosing the mean of these items, which is still in $\X$. In the new game, the average of items in $\X$ may not be in $\X$.
However, the new 
game can be interpreted as each player picks one ``meta item'' $\x'_i$ instead of $k_i$ items, 
with 
$w'_{i} = w_{i} k_{i}$ and
$\x'_{i}= \dfrac{1}{k_{i}} \displaystyle\sum_{j=1}^{k_{i}} \x^{\left(j\right)}_{i}$. 

\begin{cor}
\label{cor:Fpotential}
$F$ is also a potential game, with a potential function
\begin{align}
\phi\left(\left\{\left\{\x^{\left(j\right)}_{i}\right\}_{j=1}^{k_{i}}\right\}_{i=1}^{n}\right):=
\left\|w_{0} \x_{0} + \displaystyle\sum_{i=1}^{n} w'_{i} \x'_{i}\right\|_{2}^{2} - 2 \displaystyle\sum_{i=1}^{n} w'_{i} \t^\top_{i} \x'_{i}.
\end{align}
\end{cor}
Therefore, the new game $F$ also has at least one pure NE.
Since $\X$ is finite and the number of players is finite, $F$ is a finite game.
Although a pure NE can be similarly found through best-response dynamics, unlike the continuous case, each iteration of best-response dynamics can be costly to compute for large values of $k_{i}$ since it involves solving a variant of the subset sum problem. 
Since $F$ is finite, it clearly cannot have infinite pure NEs (contrast with corollary~\ref{cor:cvx}), but it can still have many pure NEs.

\section{Conclusion and Future Work}
We introduced the Affine Mechanism Game and its Bayesian variant, and characterized the structures of pure Nash and Bayesian Nash equilibria, including closed-form solutions under simplifying assumptions. In particular, we showed the players' extreme-exaggeration equilibrium behavior in the complete-information case and quantified the exaggeration ratio as a function of the players' influence (weight). 

The presented work has several limitations that can be enhanced by future work, including
broadening the mechanism to non-affine functions (such as trimmed mean) and characterizing the full set of NEs,
and allowing uncertainty in the players' beliefs about the mechanism (such as the weights $w_i$'s).


\section{Technical Appendices and Supplementary Material}

\subsection{Proof of Theorem \ref{thm:ctp}}

\begin{proof}  \label{proof:cptpf} 
We need to show if any player $i$ deviates from action $\x_i$ to any action $\y\in\X$,
we have
$\ell_i(\x_i,\x_{-i}) - \ell_i(\y,\x_{-i})=
\phi(\x_i,\x_{-i}) - \phi(\y,\x_{-i}).$
To this end, define an auxiliary variable $\z$ that does not depend on $\x_i$ or $\y$:
$$\z := w_0 \x_0 + \sum_{j\neq i} w_j \x_j.$$
Then 
$\ell_i(\x_i,\x_{-i})=\|w_i\x_i + (\z-\t_i)\|^2 = \|w_i\x_i\|^2 + 2w_i\x_i^\top (\z-\t_i) + \|\z-\t_i\|^2$,
and
\begin{align*}
&\ell_i(\x_i,\x_{-i})-\ell_i(\y,\x_{-i})=\|w_i\x_i\|^2 + 2w_i\x_i^\top(\z-\t_i) - \|w_i\y\|^2 - 2w_i\y^\top(\z-\t_i).
\end{align*}
On the other hand,
\begin{align*}
\phi(\x_i,\x_{-i})&=\|w_i\x_i + \z\|^2 - 2 w_i \t_i^\top \x_i - 2 \sum_{j\neq i} w_j \t_j^\top \x_j \\
&=\|w_i\x_i\|^2 + 2 w_i (\z-\t_i)^\top \x_i + \|\z\|^2 - 2 \sum_{j\neq i} w_j \t_j^\top \x_j.
\end{align*}
The last two terms do not depend on $\x_i$.  Hence
\begin{align*}
&\phi(\x_i,\x_{-i})-\phi(\y,\x_{-i}) \\
&= \|w_i\x_i\|^2 + 2 w_i (\z-\t_i)^\top \x_i - \|w_i\y\|^2 - 2 w_i (\z-\t_i)^\top \y \\
&= \ell_i(\x_i,\x_{-i})-\ell_i(\y,\x_{-i}). 
\end{align*}
\end{proof}

\subsection{Proof of Proposition \ref{prop:cvx}}

\begin{proof} 
Our potential function $\phi$ is the sum of two convex functions in $\x$ and hence convex.
Using utilities $u_i = -\ell_i$, the game has concave potential $-\phi$, and since the domain $\X^n$ is convex and $\phi$ is smooth and convex, by Theorem 1 of~\cite{neyman1997correlated} and its corollary, the set of pure Nash equilibria coincides with the minima of the potential function on $\X^n$.
\end{proof}

\subsection{Proof of Corollary \ref{cor:cvx}}

\begin{proof}  
Since the set of minima of the convex potential function on a compact domain $\X^{n}$ is non-empty, there is at least one pure Nash equilibrium. Since the set of minima of the convex potential function on a convex domain $\X^{n}$ is convex (Corollary in~\cite{neyman1997correlated} below Theorem 1), any convex combination of two distinct pure Nash equilibria is another pure Nash equilibrium. 
\end{proof}

\subsection{Proof of Theorem \ref{thm:bdy}}

\begin{proof}  \label{proof:bdypf} 
For any $i \in [n]$, 
\begin{align}
\nabla_{\x_i} \phi = 2 w_i \sum_{k=0}^n w_k \x_k - 2 w_i \t_i.
\end{align}
Suppose $\exists i, j\in [n], i\neq j: \x_i, \x_j \in \mathrm{int\;} \X$. Then 
\begin{align}
& \nabla_{\x_i} \phi = \boldsymbol{0} = \nabla_{\x_j} \phi \Rightarrow  \t_i =  \sum_{k=0}^n w_k \x_k = \t_j,
\end{align}
a contradiction.
The equation
$\hat{\mathbf{x}}= \mathbf{t}_{i^\star}$
follows from $\nabla_{\x_{i^*}} \phi = \boldsymbol{0}$ and the definition of the mechanism~\eqref{eq:receiver}.
\end{proof}


\subsection{Proof of Lemma \ref{lem:odne}}
First we prove the following lemma:
\begin{lem} [Sorted Pure Nash Equilibrium] \label{lem:sort}
For a joint action $x$, if there exist $i<j$ such that $x_j < x_i$ and $t_i < t_j$, then $x \notin \mathrm{pure NE}(G)$.
\end{lem}

\begin{proof}
By the assumption in the lemma, at least one of the following must be true:

\textbf{Case 1:} $\hat x < t_j$.
Consider the unilateral deviation of player $j$ to $x_j + \epsilon$, for $\epsilon > 0$. Let $\hat{x}'$ be the new mechanism, where $\hat{x}' = \hat{x} + w_j\epsilon$. Choose $\epsilon$ small enough such that $\hat{x} < \hat{x}' < t_j$. This will always decrease the loss of player $j$.

\textbf{Case 2:} $\hat x > t_i$.
Similarly, player $i$ unilaterally deviate to $x_i - \epsilon$, for $\epsilon > 0$. The new mechanism is $\hat{x}' = \hat{x} - w_i\epsilon$. Choose $\epsilon$ small enough such that $\hat{x} > \hat{x}' > t_i$. This will always decrease the loss of player $i$.

In either case, some player has a unilateral deviation that strictly decreases their loss, so $x\notin \mathrm{pure NE}(G)$.
\end{proof}

Now we actually prove Lemma \ref{lem:odne}
\begin{proof}
Assume that the players are sorted in ascending order by their targets. By the Lemma \ref{lem:sort}, we know that in any pure NE, the players' actions must be sorted by their targets. This shows that $\exists \underline{i}, \overline{i} \in \{0, 1, \dots, n+1\}$ s.t. $x_i = L$ for $ i \le \underline{i}$, $x_i = U$ for $i \ge \overline{i}$, because pure NE strategies of the players are sorted by their targets, and the action space is not unbounded. In other words, there could be a set of players playing $L$ in the start of the queue of players sorted by their targets, and a potential set of players playing $U$ in the end of that queue. Note that when $\underline{i} < 1$, we define that there will be no player playing $L$, and when $\overline{i} > n$, there will be no player playing $U$.

In addition, by Theorem \ref{thm:bdy}, $\forall i$ s.t. $\underline{i} < i < \overline{i}$ must satisfy $t_i = t^\star$ for some $t^\star \in [L, U]$. Therefore, $x_i = L \;\; \forall i \text{ s.t. } t_i < t^\star$, and $x_i = U \;\; \forall i \text{ s.t. } t_i > t^\star$.

By Theorem \ref{thm:bdy}, we have $\hat{x} = t^\star$, so there is
\[
\hat{x} = \displaystyle\sum_{i=0}^n w_{i} x_{i} = t^\star
\implies
\displaystyle\sum_{i : t_{i} = t^\star} w_{i} x_{i} = t^\star - w_{0} x_{0} - \displaystyle\sum_{i : t_{i} < t^\star} w_{i} L - \displaystyle\sum_{i : t_{i} > t^\star} w_{i} U
\]

\end{proof}

\subsection{Proof of Proposition \ref{prop:dec}}
\begin{proof}
The loss for any player $i \in [n]$ in $G$ can be separated this way:
\begin{align*}
\ell_i(\x) = \|w_0 \x_0 + \sum_{j=1}^n w_j \x_j - \t_i\|_2^2 = \sum_{k=1}^d (w_{0} \x_{0k} + \sum_{j=1}^n w_j \x_{jk} - \t_{ik})^2 = \sum_{k=1}^d \ell_{ik}(\x_k).
\end{align*}
We argue that $\x$ is a pure Nash equilibrium if and only if, for each coordinate $k \in [d]$, the profile $(\x_{1k}, \x_{2k}, \dots, \x_{nk})$ is a pure Nash equilibrium of $G_k$.

Assuming that the action profile $\x$ is a pure NE. Then $\forall i\in [n], \; \forall \x'_i \neq \x_i, \; \ell_i(\x_i, \x_{-i}) \le \ell_i(\x'_i, \x_{-i})$. Suppose that $\exists k \in [d], \; i \in [n], \; \x'_{ik} \ne \x_{ik} \; \text{s.t.} \; \ell_{ik}(\x_{ik}, \x_{-ik}) > \ell_{ik}(\x'_{ik}, \x_{-ik})$. Define a full-dimensional deviation $\x'$ by changing only the $k$th coordinate:
\[
x_i' = (x_{i1},\ldots,x_{i,k-1},x'_{ik},x_{i,k+1},\ldots,x_{id}).
\]
then the pure NE assumption is violated:
\[
\ell_i(\x_i, \x_{-i}) = \sum_{k=1}^d \ell_{ik}(\x_{ik}, \x_{-ik})
> \ell_{i1}(\x_1) + \dots + \ell_{ik}(\x'_{ik}, \x_{-ik}) + \dots + \ell_{id}(\x_d)
= \ell_i(\x'_i, \x_{-i}).
\]
Hence, if $\x$ is a pure NE, $\x_k$ is a pure NE for $G_k$.

Conversely, assume that for all $k \in [d]$, $\x_k = (\x_{ik}, \x_{-ik})$ is a pure NE of $G_k$. Then $\forall i, k, \forall \x'_{ik} \in \X_k$, $\ell_{ik}(\x_{ik}, \x_{-ik}) \le \ell_{ik}(\x'_{ik}, \x_{-ik})$. Summing the LHS and RHS of the inequalities gives:
\[
\ell_i(\x_i, \x_{-i}) = \sum_{k=1}^d \ell_{ik}(\x_{ik}, \x_{-ik})
\le \sum_{k=1}^d \ell_{ik}(\x'_{ik}, \x_{-ik})
= \ell_i(\x'_i, \x_{-i}) \quad \forall \x_i' \ne \x_i.
\]
This is exactly the definition of pure NE for $G$. Hence, if $\x_k$ is a pure NE for $G_k$, $\x$, which is created by stacking up all $\x_k$, is a pure NE for $G$.

Hence we have $\mathrm{pure NE}(G) = \prod_{k=1}^d \mathrm{pure NE}(G_k)$.

\end{proof}

\subsection{Proof of Proposition \ref{prop:bgbr}}
\begin{proof}
First, the expected loss of player $i$ is separable:
\begin{align*}
\mathbb{E}_{T_{-i}}\left[\ell_i(\x_i, \pi_{-i}; \t_i)\right]
&= \mathbb{E}_{T_{-i}}\!\left[\sum_{k=1}^d\left(w_0x_{0k} + w_ix_{ik}+\sum_{j\ne i} w_j\pi_{jk}(\t_{j})-t_{ik}\right)^2\Bigm| \t_i\right] \\
&= \sum_{k=1}^d\mathbb{E}_{T_{-i}}\!\left[\left(w_0x_{0k} + w_ix_{ik}+\sum_{j\ne i} w_j\pi_{jk}(\t_{j})-t_{ik}\right)^2\Bigm| \t_i\right].
\end{align*}

Fixing any $\pi_{-i}(\t_{-i})$, the best response of player $i$ given target $\t_i$ is:
\begin{align}
\x_i^* &= \arg\min_{\x_i \in \X}\sum_{k=1}^d\mathbb{E}_{T_{-i}}\left[\left(w_0x_{0k} + w_ix_{ik}+\sum_{j\ne i} w_j\pi_{jk}(\t_{j})-t_{ik}\right)^2\Bigm| \t_i\right]\label{eq:bbr}
\end{align}
When minimizing $\mathbb{E}_{T_{-i}}\left[\ell_i(\x_i, \pi_{-i}; \t_i)\right]$, term $k$ of \eqref{eq:bbr} depends on player $i$'s action only through $x_{ik}$, i.e. any change on $x_{ik'}, \;\; k'\ne k$ will not change term $k$. In addition, the domain is a product-space $\X = \prod_{k=1}^d \X_k$, so the choice of $x_{ik}$ is unaffected by the choice in any other coordinate. Thus, minimizing $\mathbb{E}_{T_{-i}}\left[\ell_i(\x_i, \pi_{-i}; \t_i)\right]$ is equivalent to minimizing each of its $d$ terms.

The distribution of targets $T_1, \dots, T_n$ are independent across players, so the expected action for each opponent is independent of $\t_i$, which means $\mathbb{E}_{T_j}[\pi_{jk}(\t_j)\mid \t_i] = \mathbb{E}_{T_j}[\pi_{jk}(\t_j)] \;\; \forall j \ne i$, given $i$ and $k$. We will use this to evaluate the derivative below.

Since term $k$ is strictly convex in $x_{ik}$, its minimum on the domain $[L_k, U_k]$ occurs at $x_{ik}=L_k$ if and only if it is monotonically increasing on $[L_k, U_k]$, which happens when
\begin{align*}
&\frac{d}{dx_{ik}}\mathbb{E}_{T_{-i}}\left[\big(w_0x_{0k} + w_ix_{ik}+\sum_{j \ne i}w_j\pi_{jk}(\t_j) - t_{ik}\big)^2\right]\big|_{x_{ik}=L_k} \ge 0 \\
&\implies t_{ik} \le w_0x_{0k} + w_iL_k + \sum_{j \ne i}w_j\mathbb{E}_{T_j}\big[\pi_{jk}(\t_j)\big].
\end{align*}

Similarly, the minimum on $[L_k, U_k]$ occurs at $x_{ik}=U_k$ if and only if term $k$ is monotonically decreasing on $[L_k, U_k]$, which happens when
\begin{align*}
&\frac{d}{dx_{ik}}\mathbb{E}_{T_{-i}}\left[\big(w_0x_{0k} + w_ix_{ik}+\sum_{j \ne i}w_j\pi_{jk}(\t_j) - t_{ik}\big)^2\right]\big|_{x_{ik}=U_k} \le 0 \\
&\implies t_{ik} \ge w_0x_{0k} + w_iU_k + \sum_{j \ne i}w_j\mathbb{E}_{T_j}\big[\pi_{jk}(\t_j)\big].
\end{align*}

The best response is an interior point of $[L_k, U_k]$ when neither condition above holds, i.e. when term $k$ of \eqref{eq:bbr} is non-monotonic on $[L_k, U_k]$. It happens when its unconstrained global minimizer $x_{ik}^*$ lies strictly within $(L_k, U_k)$,
i.e. $L_k < x^*_{ik} < U_k$.

The unconstrained best response $x_{ik}^*$ is defined as
\[
x_{ik}^* = \arg\min_{x_{ik} \in \mathbb{R}}\mathbb{E}_{T_{-i}}\left[\big(w_0x_{0k} + w_ix_{ik}+\sum_{j \ne i}w_j\pi_{jk}(\t_j) - t_{ik}\big)^2\right].
\]

To find $x_{ik}^*$, we set the derivative of that expectation to 0, and we have:
\begin{align*}
&\frac{d}{dx_{ik}}\mathbb{E}_{T_{-i}}\left[\big(w_0x_{0k} + w_ix_{ik}+\sum_{j \ne i}w_j\pi_{jk}(\t_j) - t_{ik}\big)^2\right]\big|_{x_{ik}=x^*_{ik}} = 0 \\
&\implies x_{ik}^* = \frac{1}{w_i}t_{ik} - \frac{1}{w_i}\left(w_0x_{0k} + \sum_{j \ne i}w_j\mathbb{E}_{T_j}\big[\pi_{jk}(\t_j)\big]\right).
\end{align*}

Therefore, term $k$ is minimized at
\[
\pi_{ik}^*(\t_{i}) = x_{ik}^* = 
\begin{cases}
L_k, & t_{ik} \le a_{ik},\\
\frac{1}{w_i}t_{ik} + c_{ik}, & a_{ik} < t_{ik} < b_{ik},\\
U_k, & t_{ik} \ge b_{ik},
\end{cases}
\]
where $c_{ik} = - \dfrac{1}{w_{i}} \left(w_0 x_{0k} +  \displaystyle\sum_{j \neq i} w_{j} \mathbb{E}_{T_{j}}\left[\pi_{jk}\left(t_{j}\right)\right]\right)$, $a_{ik} = w_{i} \left(L_k - c_{ik}\right)$, and $b_{ik} = a_{ik} + w_{i} \left(U_k - L_k\right) = w_{i} \left(U_k - c_{ik}\right)$.
\end{proof}

\subsection{Proof of Example \ref{eg:1dbne}}
\begin{proof}
By Proposition~\ref{prop:bgbr}, the BNE set consists of ramp profiles parameterized by $(c_1, \ldots, c_n)$ that satisfy the consistency system on our only dimension. We will solve for $c_i \; \forall i \in [n]$.

First, since $0 \le \mathbb{E}_{T_{j}}\left[\pi_{j}\left(t_{j}\right)\right] \le 1 \; \forall j, \pi_j$, for any BNE strategy profile $\{\pi^*_j(t_j)\}_{j=1}^n$, we have the following properties for $a_i, b_i$ of player $i$:
\begin{align*}
a_i &= \displaystyle\sum_{j \neq i} w_{j} \mathbb{E}_{T_{j}}\left[\pi^*_{j}\left(t_{j}\right)\right] \ge \displaystyle\sum_{j \neq i} w_{j} \cdot 0 = 0 \\
b_i &= w_i + \displaystyle\sum_{j \neq i} w_{j} \mathbb{E}_{T_{j}}\left[\pi^*_{j}\left(t_{j}\right)\right] \le w_i + \displaystyle\sum_{j \neq i} w_{j} = 1
\end{align*}
This means for strategy $\pi^*_i(t_i)$ in any BNE, it must satisfy $0 \le a_i < b_i \le 1$. We are not limiting the action space of the game, because this constraint is directly derived from existing constraints, including $\X = [0, 1]$ and the characterizations of $a_i$ and $b_i$.

Therefore we can compute $\mathbb{E}_{T_j}\big[\pi^*_j(T_j)\big]$ under the uniform prior by assuming that $0 \le a_i < b_i \le 1$ holds for this BNE strategy:
\[
\mathbb{E}_{T_j}\big[\pi^*_j(t_j)\big] = 0 \cdot a_j + 1 \cdot (1 - b_j) + \int_{a_j}^{a_j + w_j} \!\!\left(\frac{1}{w_j}t_j - \frac{a_j}{w_j}\right) dt_j = 1 + w_jc_j - \frac{w_j}{2}.
\]
Substituting into the consistency system $c_{i} = - \dfrac{1}{w_{i}}\displaystyle\sum_{j \neq i} w_{j} \mathbb{E}_{T_{j}}\left[\pi^*_{j}\left(t_{j}\right)\right]$ yields the linear system
\[
w_ic_i + \sum_{j \ne i} w_j^2c_j = -\sum_{j \ne i}\big(w_j - \frac{w_j^2}{2}\big), \quad \forall i \in [n].
\]
Let the matrix representing the LHS weights be $M$, with
$M_{ij} = 
\begin{cases}
w_i & \text{if}\; j=i \\
w_j^2 & \text{if}\; j \ne i
\end{cases}$
subtracting row $1$ from all other rows, and then applying $ \text{row}\; 1 - \frac{w_i^2}{w_i - w_i^2} \cdot \text{row}\; i, \; \forall i\in \{2, \dots, n\}$ makes a lower triangular matrix with non-zero diagonal values. This means that $M$ is full rank and there is a unique solution to our system of equations. $c_i = -\frac{1-w_i}{2w_i} \; \forall i \in [n]$ is the unique solution here, with verification omitted.

If we plug $\{c_i\}_{i=1}^n$ back into the ramp function, the unique BNE strategy has $a_i = \frac{1-w_i}{2}, b_i = \frac{1+w_i}{2}$, and $c_i = - \frac{1-w_i}{2w_i}$, which matches our claim in Example~\ref{eg:1dbne}.
\end{proof}

\subsection{Proof of Proposition \ref{prop:wdse}}

\begin{proof}
Given an arbitrary joint action from other players $\x_{-i}$, player $i$'s best response is
\begin{align*}
BR(\x_{-i}) &:= \argmin_{\x_i \in \X} \ell_i(\x_i, \x_{-i})  = \argmin_{\x_i \in \X} - \t_i^\top ( w_0 \x_0 + \sum_{j=1}^n w_j \x_j )\\
&= \argmin_{\x_i \in \X} - w_i \t_i^\top \x_i + \mathrm{const} = \x_i^*.
\end{align*}
The best response is independent of $\x_{-i}$.
\end{proof}

\subsection{Proof of Corollary \ref{cor:Fpotential}}

\begin{proof}
The proof is identical to the one for Theorem \ref{thm:ctp}, since the loss functions are identical, only on a different domain.
\end{proof}

\vspace{3em}


\begingroup
\let\clearpage\relax
\let\newpage\relax
\bibliography{big}
\endgroup

\end{document}